\documentclass[9pt,shortpaper,twoside,web]{ieeecolor}

\usepackage{amsthm}

\usepackage{graphics} 
\usepackage{times} 
\usepackage{amsmath,amssymb,amsfonts}
\usepackage{mathrsfs}
\usepackage{algorithmic}
\let\labelindent\relax
\usepackage{enumitem}
\usepackage{tcolorbox}
\usepackage{graphicx}
\usepackage{textcomp}
\usepackage{varioref}
\usepackage{import}
\usepackage{dsfont}
\usepackage{framed,xcolor}
\usepackage{color}
\usepackage{multirow}
\usepackage{float}
\usepackage{mathtools}
\usepackage{booktabs}
\usepackage{soul}
\usepackage{amsbsy}
\usepackage[bottom]{footmisc}
\usepackage{lineno}
\usepackage{microtype}
\usepackage{comment}

\newtheorem{Example}{Example}

\usepackage{bm} 
\usepackage{tikz}
\usetikzlibrary{shapes.geometric, arrows.meta, positioning}
\usepackage{standalone}
\usepackage[T1]{fontenc}
\usepackage{orcidlink}
\def\logoname{LOGO-generic-web}
\definecolor{subsectioncolor}{rgb}{0,0.541,0.855}
\usepackage{generic} 

\newtheorem{theorem}{Theorem}

\newtheorem{lemma}[theorem]{Lemma}

\newtheorem{proposition}[theorem]{Proposition}
\newtheorem{assumption}[theorem]{Assumption}

\newcommand{\nat}{\mathbb{N}}

\usepackage{mathrsfs}  
\usepackage{color}

\usepackage{hyperref}

\title{On Certificates for Almost Sure Reachability \\ in Stochastic Systems}
\author{Arash Bahari Kordabad, Rupak Majumdar, Harshit Jitendra Motwani, and Sadegh Soudjani \thanks{The authors are affiliated with the Max Planck Institute for Software Systems, Kaiserslautern, Germany. E-mail: {\tt\small\{arashbk, rupak, hmotwani, sadegh\}@mpi-sws.org.}}}

\begin{document}
\maketitle
\begin{abstract}
Almost sure reachability refers to the property of a stochastic system whereby, from any initial condition, the system state reaches a given target set with probability one. In this paper, we study the problem of certifying almost sure reachability in discrete-time stochastic systems using drift and variant conditions. While these conditions are both necessary and sufficient in theory, computational  approaches often rely on restricting the search to fixed templates, such as polynomial or quadratic functions. We show that this restriction compromises completeness: there exists a polynomial system for which a given target set is almost surely reachable but admits no polynomial certificate, and a linear system for which a neighborhood of the origin is almost surely reachable but admits no quadratic certificate. We then provide a complete characterization of reachability certificates for linear systems with additive noise. Our analysis yields conditions on the system matrices under which valid certificates exist, and shows how the structure and dimension of the system determine the need for non-quadratic templates. Our results generalize the classical random walk behavior to a broader class of stochastic dynamical systems.
\end{abstract}
\IEEEpeerreviewmaketitle
\section{Introduction}
\label{sec:intro}
The reachability problem, determining whether a dynamical system's state eventually enters a desired target set, is a foundational question across multiple disciplines such as control theory~\cite{lygeros2004reachability,bujorianu2012stochastic} and formal verification~\cite{baier2008principles}. While classical methods for deterministic systems provide strong guarantees using tools such as Lyapunov functions and barrier certificates~\cite{blanchini2008set, ames2019control}, the introduction of randomness makes the analysis significantly more challenging~\cite{prajna2007framework, lavaei2022automated}. In stochastic settings, probabilistic versions of tools such as barrier certificates have been developed to provide safety and reachability guarantees in a probabilistic sense~\cite{kordabad2024control,santoyo2021barrier}. However, reachability guarantees that hold \emph{with probability one}, a property known as \emph{almost sure reachability}~\cite{chatterjee2016optimal}, are typically more desirable.

Almost sure reachability is a critical property in many domains, from verifying the reliability of stochastic control systems~\cite{jafarpour2025probabilistic} to analyzing the termination of probabilistic programs~\cite{mciver2017new,abate2021learning}. In computer science, this property is closely related to \emph{almost sure termination}, where one asks whether a probabilistic program will eventually terminate with probability one for all initial configurations. This question is especially relevant for verifying probabilistic programs with nondeterminism. Reachability analysis has been studied in the context of termination of probabilistic programs, e.g., in~\cite{ferrer2015probabilistic}, and necessary and sufficient proof rules have recently been developed~\cite{majumdar2025sound}.

 Recent theoretical work~\cite{majumdar2024necessary} has introduced a necessary and sufficient condition for certifying almost sure reachability in discrete-time stochastic systems. This framework involves the construction of two functions: 1) a \emph{drift function}, that is, a radially unbounded function that does not increase in expectation on every state transition outside of a compact set, and 2) a \emph{variant function}, which is positive outside the target set and decreases on each transition by a minimal amount with a minimal probability within each sublevel set of the drift. In addition, the variant is bounded above within each sublevel set of the drift. In~\cite{majumdar2024necessary}, it has been shown that the presence of a drift and a variant is necessary and sufficient for the target set to be reached almost surely.

While the drift and variant framework provides a complete theoretical characterization of almost sure reachability, it remains unclear when such certificates exist within specific function classes. In particular, we show that restricting the search to fixed templates can result in the loss of completeness. We construct a polynomial system that almost surely reaches a given target set but does not admit any polynomial certificate, as well as a linear system that almost surely reaches a neighborhood of the origin yet lacks a quadratic certificate.

\textbf{Contribution.} Considering the limitations mentioned above, we turn to linear systems with additive stochastic disturbances and provide a complete characterization of when almost sure reachability holds e.g., based on spectral radius of the system matrices. Similar spectral conditions have appeared in the stability literature~\cite{edition1999linear}. More precisely, we identify conditions, based on the spectral radius, Jordan structure, number of modulus one, and excitation directions from the noise, under which valid certificates exist and when non-quadratic forms become necessary. This provides a generalization of classical random walk behavior to a broader class of stochastic dynamical systems, including those with structured dynamics and non-identically distributed noise.

The main contributions of this paper are as follows:  
\begin{enumerate}[label=(\roman*)]
    \item We show that restricting drift or variant functions to fixed templates, such as polynomial or quadratic forms, may break completeness even for simple examples.
    
    \item We provide a complete characterization of almost sure reachability for linear systems with additive noise based on the system matrices, identifying when valid certificates exist and when non-quadratic forms are required. Moreover, we explicitly construct non-quadratic certificates based on logarithmic functions.
    \item We generalize the results on the classical random walk to a broader class of stochastic dynamical systems.
\end{enumerate}

Our results expose both fundamental limitations of template-based certificate synthesis and provide precise structural conditions under which reachability certificates exist for linear systems. By connecting properties of linear systems to the existence of suitable certificates, the paper contributes a unified theoretical framework for reasoning about almost sure reachability in linear stochastic systems.

\textbf{Outline.} The rest of the paper is organized as follows. Section~\ref{sec:preliminaries} introduces the preliminaries and background on almost sure reachability. In Section~\ref{sec:polynomial-systems}, we present counterexamples showing that restricting certificates to polynomial or quadratic forms may lead to the loss of completeness. Section~\ref{sec:linear-system} focuses on characterizing almost sure reachability for linear systems with corresponding certificates. Finally, Section~\ref{sec:conclusion} concludes the paper and outlines directions for future work.

\textbf{Notation.}  
We denote the set of real, non-negative real and positive real numbers, by $\mathbb{R}$, $\mathbb{R}_{\geq 0}$ and $\mathbb{R}_{>0}$, respectively. The set of non-negative (positive) integers by $\mathbb{N}_{\geq 0}$ ($\mathbb{N}$), and the set of complex numbers by $\mathbb{C}$. The Euclidean norm of a vector $x \in \mathbb{R}^n$ is denoted by $\|x\|$ and the Frobenius norm is denoted by $\|x\|_F$. A square symmetric matrix $M$ is positive (semi-) definite when $x^\top M x>0$ ($x^\top M x\geq 0$) for all $x\neq 0$ and it is denoted by $M \succ 0$ ($M \succeq 0$). The expectation operator is denoted by $\mathbb{E}[\cdot]$, and $\mathbb{P}(\cdot)$ denotes probability. For random vectors $x$ and $y$, the covariance is defined as $ \operatorname{Cov}(x, y) := \mathbb{E}[(x - \mathbb{E}[x])(y - \mathbb{E}[y])^\top]$ and the variance is $\operatorname{Var}(x) := \operatorname{Cov}(x, x)$.
The indicator function of a set $S$ is denoted by $\mathbf{1}_S(x)$, that is $\mathbf{1}_S(x)=1$ if $x\in S$ and $\mathbf{1}_S(x)=0$ otherwise. The dimension of a vector space is denoted by $\dim(\cdot)$. The spectral radius of a square matrix $A$ is denoted by $\rho(A)$ and it is the maximum of the absolute values of its eigenvalues. A set \( S \subseteq \mathbb{R}^n \) is called \emph{open} if for every point \( x \in S \), there exists a neighborhood of \( x \) entirely contained in \( S \). A set is \emph{closed} if its complement is open. A set $S \subset \mathbb{R}^n$ is \emph{bounded} if there exists an $M\in \mathbb{R}_{\geq 0}$ such that $\|x\| \leq M$ for all $x \in S$. A set is said to be \emph{compact} if it is closed and bounded.

\section{Preliminaries and Problem Statement}
\label{sec:preliminaries}
We consider a discrete-time polynomial stochastic system (dt-PSS) as a tuple $\Sigma=(X,W,w, f)$, where $X\subseteq\mathbb{R}^n$ is the state space of the system,
$W\subseteq\mathbb{R}^m$ is the uncertainty space,
$w:=\{w_k:\Omega\rightarrow W,k\in\nat_{\geq 0}\}$ is a sequence of independent and identically distributed (i.i.d.) random variables defined on sample space $\Omega$. Each $w_k$ takes values in $W$ and is sampled according to a probability measure $\mathbb{P}_w$,
and the map $f:X\times W\rightarrow X$ is a polynomial function in both arguments that characterizes the state evolution of the system according to:
	\begin{equation}
	x_{k+1}=f(x_k,w_k), \quad x_0: \text{initial condition}
	\label{eq:state_evolution}
	\end{equation}
where $x_k\in X$ is the state at time $k\in\nat_{\geq 0}$, and $w_k\in W$ is a random variable of zero mean with finite covariance matrix $\Sigma_w \succ 0$, and its (possibly unbounded) support contains an open ball centered at the origin. A special case of~\eqref{eq:state_evolution} is a linear system:
\begin{equation}\label{eq:linear_system}
    f(x_k,w_k) = A x_k +B w_k,
\end{equation}
where $A \in \mathbb{R}^{n \times n}$ and $B \in \mathbb{R}^{n \times m}$ are constant system matrices.

\textbf{Almost Sure Reachability Problem.}
\label{prob:asr}
Given a dt-PSS $\Sigma$ and a bounded open target set $G \subseteq X$, determine whether for all initial conditions $x_0 \in X$, the system trajectory $\{x_k\}_{k=0}^\infty$ hits the target set $G$ for almost all possible noise realizations $\{w_k\}^\infty_{k=0}$, i.e., 
\begin{equation}\label{eq:reach}
    \forall x_0\in X,\quad \mathbb{P}\left( \exists\, k \in \mathbb{N}_{k\geq 0} : x_k \in G  \right) = 1.
\end{equation} 
That is, does the system reach the target set $G$ almost surely from every initial condition $x_0$?

To characterize almost sure reachability, we follow the framework proposed in~\cite{majumdar2024necessary}, which provides a pair of conditions that are both necessary and sufficient for ensuring that the trajectories reach a given target set \( G \) with probability one. 
The theorem in~\cite{majumdar2024necessary} applies to more general discrete-time stochastic systems that are \emph{weak Feller}---we note that every dt-PSS is weak Feller.
We recall these criteria next.

\noindent\textbf{V1: Drift Criterion.} There exists a drift function $V:X\rightarrow\mathbb{R}_{\geq 0}$ with $\lim_{\|x\|\rightarrow\infty}V(x)=\infty$ and a compact set $C\subseteq X$ satisfying,
\begin{equation}
\label{eq:drift}
    \Delta V(x) :=\mathbb{E}\left[ V(f(x, w))\,|\,x \right] - V(x) \le 0,\quad \forall x\in C^c,
\end{equation}
where $C^c$ is the complement of $C$.

\noindent\textbf{V2: Variant Criterion.} For a function $V$ satisfying the drift criterion \textbf{V1}, there exists a function $U : X \to \mathbb{R}$ called the \emph{variant}, and three supporting functions $H: \mathbb{R}_{>0} \to \mathbb{R}$, $\delta : \mathbb{R}_{>0} \to \mathbb{R}_{>0}$, and $\epsilon : \mathbb{R}_{>0} \to \mathbb{R}_{>0}$ such that for all $r\in \mathbb{R}_{>0}$ and $x \in X$, the implication 
$
    V(x) \leq r \implies U(x) \leq H(r)
$
holds, and 
$$
\mathbb P(U(f(x,w))-U(x)\le -\delta(r))\ge \epsilon(r),
$$
for all $x$ satisfying $V(x) \leq r$ and $U(x)> 0$.

\begin{theorem}[Almost sure reachability~\cite{majumdar2024necessary}]
\label{thm:nece_suff}
    For a dt-PSS $\Sigma$ and open bounded target set $G$, if there exists a function $V$ satisfying criterion \textbf{V1} and a variant $U$ satisfying criterion \textbf{V2} associated with $V$, such that
    \begin{equation}\label{eq:G:U}
     G \supset \{x \in X \mid U(x) \leq 0\},  
    \end{equation}
then~\eqref{eq:reach} holds. Moreover, if \eqref{eq:reach} holds, then there are functions $V$ and $U$ that satisfy \textbf{V1} and \textbf{V2}, and such that~\eqref{eq:G:U} holds.
\end{theorem}

Conditions \textbf{V1} and \textbf{V2} play complementary roles in ensuring almost sure reachability. The drift condition \textbf{V1} guarantees that system trajectories do not diverge to infinity with positive probability; instead, they remain bounded almost surely. This is achieved by enforcing a supermartingale-like decrease of the function \( V \) outside a compact set \( C \), which can be arbitrarily large. The variant condition \textbf{V2} ensures that there is enough stochastic excitation to make progress toward the target set \( G \) with a strictly positive probability. Note that checking criteria \textbf{V1} and \textbf{V2} for given \( V \) and \( U \) does not need the probability measure on the infinite state trajectory and can be verified using only the one-step transition dynamics of the system and the distribution of the disturbance \( w \). However, this result provides only an existential guarantee, it does not provide guidance on how to construct such functions, nor does it indicate the class of functions in which to search. In the next section, we address this question by examining the limitations of template-based certificates. Specifically, we construct systems that almost surely reach a target set but do not admit any polynomial or quadratic function that satisfies the drift condition. These counterexamples reveal that restricting the search space of certificates can lead to incompleteness, motivating the need for deeper structural analysis.

\section{The Limitations of Polynomial Certificates}
\label{sec:polynomial-systems}

While the drift and variant conditions presented in the previous section provide a necessary and sufficient condition for verifying almost sure reachability, a natural question arises in practice: can we restrict our attention to a specific function class—e.g., polynomial or quadratic functions—and still retain the necessity direction of the result? This would be highly desirable in practice, enabling tractable synthesis via tools such as sum-of-squares optimization and semidefinite programming.

In the following example, we construct a dt-PSS for which almost sure reachability to a bounded open target set holds, yet no polynomial function satisfies the drift condition~\eqref{eq:drift}. This example is inspired by the continuous-time deterministic dynamics in~\cite{ahmadi2011globally}, which showed that a polynomial system could be globally stable without admitting a polynomial Lyapunov function. Our construction adapts that example, with some modifications, to a discrete-time setting and incorporates stochastic dynamics. While randomness is not essential for the counterexample, it reflects the stochastic nature of the systems studied in this paper. In contrast to the continuous-time analysis, which uses a gradient-based Lyapunov condition, our contradiction arises from the failure of any polynomial function to satisfy the drift inequality in \textbf{V1}. These changes make the example both aligned with the objectives of this work and sufficiently distinct, highlighting that the limitations of polynomial certificates persist even in stochastic, discrete-time settings under almost sure reachability.

\begin{Example} Consider the following system:
\begin{equation}\label{eq:examp1}
f(x_k, w_k) := 
\begin{bmatrix}
\frac{1}{2} \, \xi_k (1 + \eta_k + w_k) \\
\frac{1}{2} \, \eta_k
\end{bmatrix}, \quad x_k:=  \begin{bmatrix}
 \xi_k  \\
\eta_k
\end{bmatrix},
\end{equation}
where $w_k$ is uniformly distributed in the interval $\left[-{1}, {1}\right]$ and $\xi_k, \eta_k> 0$. One can verify that drift and variant functions of the form
\begin{equation*}
V(\xi,\eta)=U(\xi,\eta)+2=\ln(1+\xi)+\eta^2,
\end{equation*}
satisfy \textbf{V1} and \textbf{V2} for a compact set $C$ and bounded open target set $G=\{(\xi,\eta)\in \mathbb{R}^2\,|\,0<\xi, \eta< 1\}$. Therefore, from Theorem~\ref{thm:nece_suff}, the system almost surely reaches $G$. 

We show that there is no polynomial drift function for this system. The initial state is taken as $[\xi_0,\eta_0]=[2^{i},2^{i}u]$ with a fixed $u \geq 1$ and an integer $i\in \mathbb{N}$. Note that this particular class of initial conditions is chosen to simplify the computations and to make the contradiction argument for the nonexistence of a polynomial drift function more transparent. The system trajectories with initial conditions $[\xi_0,\eta_0]=[2^{i},2^{i}u]$ admit a closed-form solution as
\begin{equation*}
\xi_k=2^{i} \prod_{n=0}^{k-1} \frac{1}{2}\left(1+u 2^{i}\left(\frac{1}{2}\right)^n+w_n\right),\quad
\eta_k=u 2^{i}\left(\frac{1}{2}\right)^k.
\end{equation*}
We now consider a specific time step and derive bounds for the state at that time. The time it takes for the trajectory to cross the line $\eta = u\geq 1$ is denoted by $k^\star$ and is obtained as
$\eta_{k^\star}=u 2^{i}\left(\frac{1}{2}\right)^{k^\star}= u
$  which implies $k^\star=i$.
Next, we derive the bounds on $\xi_{k^\star}$ at this crossing time. Substituting $k^\star=i$ into the expression for $\xi_{k^\star}$ we obtain the following bounds:
\begin{align*}
    \xi_{k^\star} \ge & \prod_{n=0}^{i-1} u 2^{i}\left(\frac{1}{2}\right)^n= u^{i} 2^{\frac{1}{2}i(i+1)},\\
     \xi_{k^\star} \leq & \prod_{n=0}^{i-1} \left( 2+ u 2^{i}\left(\frac{1}{2}\right)^n\right) \leq \prod_{n=2}^{i+1}  u 2^{n} = u^{i} 2^{\frac{1}{2}i(i+3)}.
\end{align*}
We now show by contradiction that no polynomial \( V \) satisfies criterion \textbf{V1}. Suppose that there exists a polynomial \( V \) that satisfies \textbf{V1}. Note that since $u$ is an arbitrary fixed parameter, we can choose it such that $(\xi_k,\eta_k)\in C^c$ for all $k\leq k^\star$, e.g., $u>\sup_{x\in C}\|x\|$, since $\eta_k$ is decreasing. Then we have,
\begin{align}
   &\mathbb{E} [V(\xi_{k+1},\eta_{k+1})] - V(\xi_k,\eta_k) \le 0,\quad \forall (\xi_k,\eta_k)\in C^c  \nonumber\\
   \Rightarrow\,\, &\mathbb{E} [V(\xi_{k+1},\eta_{k+1})] - V(\xi_k,\eta_k) \le 0,\quad \forall k\le k^\star \nonumber\\
   \Rightarrow\,\, & \mathbb{E} [V(\xi_{k^\star},\eta_{k^\star})] - V(\xi_0,\eta_0) \le 0,\nonumber\\
   \Rightarrow\,\, & \min_{\xi_{k^\star}\in[ u^{i} 2^{\frac{1}{2}i(i+1)} ,u^{i} 2^{\frac{1}{2}i(i+3)}]} V(\xi_{k^\star},u)\le  V(2^{i} ,2^{i} u),\label{eq:final_inequality0}
\end{align}
because the minimum of a function is not greater than its expectation. Any positive polynomial $V(\xi,\eta)$, with $\lim_{\|[\xi,\eta]\|\rightarrow\infty}V(\xi,\eta)=\infty$, is non-decreasing for all $\xi\geq L$ for a sufficiently large $L$ and fixed $\eta$. Then for any $i$ satisfying $u^{i} 2^{\frac{1}{2}i(i+1)}>L$, we have,
\begin{equation*}
V(u^{i} 2^{\frac{1}{2}i(i+1)},u) \leq \min_{\xi_{k^\star}\in [u^{i} 2^{\frac{1}{2}i(i+1)} ,u^{i} 2^{\frac{1}{2}i(i+3)}]} V(\xi_{k^\star},u).
\end{equation*}
Combining this with~\eqref{eq:final_inequality0} yields:
\begin{equation}\label{eq:final_inequality}
    V(u^{i} 2^{\frac{1}{2}i(i+1)},u) \leq  V(2^{i} ,2^{i} u).
\end{equation}
Note that $\lim_{\|[\xi,\eta]\|\rightarrow\infty}V(\xi,\eta)=\infty$, hence it is not a constant function of $\xi$. Take the polynomial $V$ of finite degree $d$, expressed as $V(\xi,\eta)=\sum_{\ell=0}^d \sum_{j=0}^{d-\ell} a_{\ell j} \xi^{\ell} \eta^{j},$ with at least one $a_{\ell j}\neq 0$ for some $\ell>0$. Then \eqref{eq:final_inequality} reads as
\begin{equation}\label{eq:ineq:pol:exm}
\sum_{\ell=0}^d \sum_{j=0}^{d-\ell} a_{\ell j} \left(u^{i} 2^{\frac{1}{2}i(i+1)}\right)^{\ell} u^{j} 
\leq 
\sum_{\ell=0}^d \sum_{j=0}^{d-\ell} a_{\ell j}  \left(2^{i}\right)^{\ell} \left(2^{i} u\right)^{j}. 
\end{equation}
Dividing both sides of the above inequality by $2^{i(d+1)}$ and taking the limit of both sides as $i\rightarrow\infty$, we obtain
\begin{subequations}\label{eq:lim:pol:exm}
  \begin{align}
&\lim_{i\rightarrow\infty} \sum_{\ell=0}^d \sum_{j=0}^{d-\ell} a_{\ell j} u^{i\ell+j} 2^{\frac{1}{2}i(i+1)\ell -i(d+1)}=\infty
 \\
&\lim_{i\rightarrow\infty} \sum_{\ell=0}^d \sum_{j=0}^{d-\ell} a_{\ell j} u^{j}  2^{i(\ell+j-d-1)}=0,  
\end{align}  
\end{subequations}
since $\ell+j\leq d$. The contradiction between~\eqref{eq:ineq:pol:exm} and \eqref{eq:lim:pol:exm} establishes that no polynomial $V$ satisfies criterion \textbf{V1}. 
\end{Example} 
\smallskip

We now turn to a second example: a \emph{one-dimensional random walk}, which, despite being governed by linear dynamics and almost surely reaching a neighborhood of the origin, does \emph{not} admit a quadratic drift function satisfying criterion~\textbf{V1}.

\begin{Example} We consider~\eqref{eq:linear_system} with $A=B=1$ and $X=\mathbb{R}$, where $w_k$ is uniformly distributed in the interval $\left[-{1}, {1}\right]$ and the target set being $G=\{x\in\mathbb{R}\,|\,|x|<2\}$. It is well known that this system reaches any ball around the origin almost surely (see e.g., Chung Fuchs Theorem~\cite{chung1951distribution}). To verify this using Theorem~\ref{thm:nece_suff}, we can show that the drift and variant functions in the form of
\[
V(x) = |x|,\qquad U(x)=|x|-1,
\]
satisfy \textbf{V1} and \textbf{V2} with $C=[-1,1]$, $H(r)=r-1$, and constants $\delta\in(0,1)$ and $\epsilon={(1-\delta)}/{2}$. However, suppose we restrict the drift function to be quadratic, \(V(x) = ax^2+bx+c\). Note that $a>0$, because $V(x)$ must be nonnegative for large enough $|x|$.
Then
\begin{align*}
    \Delta V(x) = &\mathbb{E}[a(x + w)^2+b(x+w)+c] - (ax^2+bx+c) \\=&a \mathbb{E}[w^2] + (2xa+b) \mathbb{E}[w] = \frac{a}{3},
\end{align*}
since \(\mathbb{E}[w] = 0\) and \(\mathbb{E}[w^2] = \frac{1}{3}\). Thus, \(\Delta V(x) > 0\) for all \(x \neq 0\), violating the drift condition \textbf{V1}.
\end{Example}
\smallskip 
This demonstrates that no quadratic drift function can certify almost sure reachability in this example. This example underscores a key limitation of template-based methods: even simple linear systems may lack quadratic certificates, and completeness can be lost when restricting to a specific class of functions.

These examples illustrate that restricting certificate templates, such as to polynomials or quadratics, can lead to incompleteness, even for simple systems with almost sure reachability. We therefore shift our focus from assuming specific function classes to analyzing the structural properties of the system itself. In particular, we study linear systems with additive noise and investigate under what conditions these systems admit valid certificates satisfying the drift and variant conditions. Rather than assuming that a quadratic or other template-based function will suffice, we provide a principled analysis based on the properties of the system. This approach allows us to identify when quadratic certificates can be constructed and when alternative, non-polynomial forms—such as logarithmic functions—are necessary. Our goal is to characterize almost sure reachability of linear systems with additive noise to a bounded open target set $G$ containing the origin and identifying conditions under which valid certificates can be constructed. Flowchart~\ref{fig:flowchart} illustrates the different cases of almost sure reachability depending on the properties of the linear system.

\begin{figure}[t]
\hspace{0cm}
 \tikzset{
	base/.style = {font=\small},  
	start/.style = {base, rectangle, rounded corners, draw, fill=violet!30, 
		text width=2.5cm, align=center, minimum height=1cm},
	decision/.style = {base, diamond, aspect=2, draw, fill=yellow!30, inner sep=2pt, 
		align=center, minimum width=2.5cm},
	outcomeYes/.style = {base, rectangle, rounded corners, draw, fill=green!30, 
		align=center, text=blue, minimum width=2.8cm, minimum height=1cm},
	outcomeNo/.style = {base, rectangle, rounded corners, draw, fill=green!30, 
		align=center, text=red, minimum width=2.8cm, minimum height=1cm},
	outcomeInc/.style = {base, rectangle, rounded corners, draw, fill=green!30, 
		align=center, text=purple, minimum width=2.8cm, minimum height=1cm},
	arrow/.style = {thick, -Stealth}  
}
	\begin{tikzpicture}[node distance=1.5cm]
		\node   (startLabel) [yshift=-2cm]{System~\eqref{eq:linear_system} and $G$};
		\node (dec2) [decision, below=of startLabel,yshift=+0.5cm] {$\rho(A) = 1$};
		\node (outcomeNot) [outcomeNo, below=of dec2,xshift=0cm] {Not almost sure reachable};
		\node (dec1) [decision, right=of outcomeNot,xshift=0.4cm] {$\rho(A) < 1$};
		\node (dec3) [decision, left=of outcomeNot,xshift=0.7cm] 
		{All Jordan blocks of\\ $|\lambda_i| = 1$ are of size\\one};
           		\node (dec4) [decision, below=of outcomeNot] 
		{Number of eigenvalues\\ with $|\lambda_i| = 1$ is at most $2$};
        \node (dec5) [decision, below=of dec3,yshift=0.068cm] 
		{$B$: full rank};
		\node (outcomeReach) [outcomeYes, below=of dec1,yshift=-0.7cm] {almost sure reachable};
     \node (endlabel) [below=of dec5, yshift=0.7cm]{Not known};
		\draw[arrow] (startLabel.south) -- (dec2.north); 
		\draw[arrow] (dec2.west) --node[near start, above]{Yes} ++(-3.65,0)--  (dec3.north);
		\draw[arrow] (dec2.east) --node[near start, above]{No}  ++(3.62,0)-- (dec1.north);
\draw[arrow] (dec1.south) --  node[near start, right]{Yes}(outcomeReach.north);
		\draw[arrow] (dec1.west) -- node[near start, above]{No} (outcomeNot.east);
		\draw[arrow] (dec3.south) --node[midway, left]{Yes}(dec5.north);        
		\draw[arrow] (dec3.east) -- node[near start, above]{No} (outcomeNot.west);
           \draw[arrow] (dec4.east) -- node[near start, below]{Yes} (outcomeReach.west);
                      \draw[arrow] (dec4.north) -- node[near start, left]{No}(outcomeNot.south);
\draw[arrow] (dec5.east) -- node[near start, above]{yes} (dec4.west);
\draw[arrow] (dec5.south) -- node[near start, left]{No} (endlabel.north);         
	\end{tikzpicture} 
  \caption{Flowchart illustrating conditions for almost sure reachability of stochastic linear systems with additive noise based on the spectral radius and Jordan structure of $A$, as well as rank of $B$.}
  \label{fig:flowchart}
\end{figure}
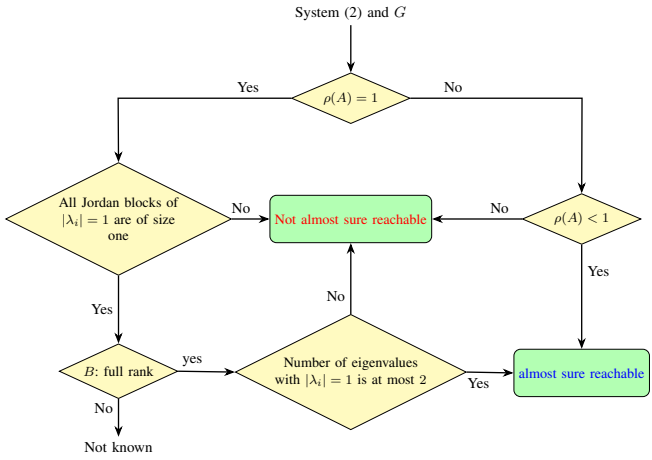

\section{Almost Sure  Reachability for Linear Systems}
\label{sec:linear-system}
In this section, we present a systematic analysis of almost sure reachability for stochastic linear systems in the form of~\eqref{eq:linear_system}. 
We provide a characterization for almost sure reachability in terms of the spectral radius, Jordan structure, and unit-modulus eigenvalues of $A$, as well as rank of $B$ based on Theorem~\ref{thm:nece_suff}. First, we present the results for $\rho(A)\neq 1$ in the following theorem.

\begin{theorem}\label{theorem:linear}
Almost sure reachability for the linear system~\eqref{eq:linear_system} holds for any bounded open target set containing the origin, and quadratic drift and variant functions satisfying \textbf{V1} and \textbf{V2} exist if $\rho(A)<1$. Moreover, the system does not almost surely reach any bounded open target set if $\rho(A)>1$.
\end{theorem}
\begin{proof} \textbf{Part A : \boldmath$\rho(A) < 1$}. 
 It is known that for $A$ with $\rho(A) < 1$, there exists a matrix $Q$ such that $A^\top Q A=Q-I$ (see e.g., Theorem 5. D5 in~\cite{edition1999linear}). Then for any $\varepsilon\in [0,1]$, $
 A^\top Q A \preceq Q - \varepsilon I$. This implies that for all $x \in \mathbb{R}^n$,
\begin{equation*}
   x^\top (A^\top Q A - Q) x \leq -\varepsilon x^\top x.
\end{equation*}
Since $\Sigma_w\succ 0$ is a finite matrix, we define the compact set $C$ as
\begin{equation}\label{eq:compact:lin}
C := \left\{ x \in \mathbb{R}^n \;\middle|\; x^\top x \leq \frac{1}{\alpha} \operatorname{tr}(B^\top Q B\, \Sigma_w) \right\},
\end{equation}
and choose an $\alpha > 0$ such that $\varepsilon \geq \alpha$. Then the following holds:
\begin{align*}
    x^\top x &> \frac{1}{\alpha} \operatorname{tr}(B^\top Q B \Sigma_w)
\quad \Rightarrow \nonumber \\ 
\varepsilon x^\top x & > \frac{\varepsilon}{\alpha} \operatorname{tr}(B^\top Q B \Sigma_w) \geq \operatorname{tr}(B^\top Q B \Sigma_w), 
\end{align*}
for all $x \in C^c$. Thus,
\begin{equation*}
    x^\top (A^\top Q A - Q)x \leq -\varepsilon x^\top x < -\operatorname{tr}(B^\top Q B \Sigma_w),
\end{equation*}
implies that the following holds for all $x\in C^c$ and $V(x)=x^\top  Q x$:
\begin{align*}
    &\mathbb{E}[V(Ax+Bw)]-V(x)\leq 0\Leftrightarrow \nonumber\\ & x^\top (A^\top Q A - Q) x + \operatorname{tr}(B^\top Q B\, \Sigma_w)\leq 0,
\end{align*} 
 with the compact set $C$ in~\eqref{eq:compact:lin}, and therefore condition~\textbf{V1} is satisfied.  

\medskip
For the variant function, first, recall the weighted vector norm $\|x\|^2_Q = x^\top Q x$ and the corresponding induced matrix norm $\|A\|_Q = \sup_{x \neq 0} {\|Ax\|_Q}/{\|x\|_Q}$ for a positive definite $Q$ satisfying $A^\top Q A= Q-I$. Thus $A^\top Q A \prec Q$ and it follows that $r_0 := \|A\|_Q^2 = \lambda_{\mathrm{max}}(Q^{-1}A^\top Q A) < 1$. For a bounded open target set $G$ containing the origin, we can under-approximate the set $G$ by a sublevel set of the form $\{x \mid \|x\|^2_Q < 2{b}\}\subset G$ for some $b>0$, i.e., $\|x\|_Q^2\geq 2b$ for all $x\in G^c$. Define the variant function as $U(x) = x^\top Q x - b$. It implies $\{x\mid U(x)\leq 0\}\subset G$.  Let $\delta = (1 - r_0)b > 0$. For all $w$ in the set $\{w \mid \|w\|_Q^2 \leq {(1 - r_0)b}/{\|B\|_Q^2}\}$, which has positive measure, we have,
\begin{align*}
    &U(Ax+Bw) - U(x) =  \|Ax + Bw\|_Q^2 - \|x\|_Q^2\\
    &\leq \|A\|_Q^2 \|x\|_Q^2 + \|B\|_Q^2 \|w\|_Q^2 - \|x\|_Q^2 \\
    &= (r_0 - 1)\|x\|_Q^2 + \|B\|_Q^2 \|w\|_Q^2 \leq 2(r_0 - 1)b + (1 - r_0)b\\
    &= -(1 - r_0)b = -\delta,
\end{align*}
for all $x\in G^c$. Thus, condition \textbf{V2} holds.
\smallskip

\textbf{Part B : \boldmath$\rho(A) > 1$}. We provide the proof through four steps:

\smallskip
\textbf{Step 1:} We first show that there exists a real $v_0 \in \mathbb{R}^n$ such that for every initial condition $x_0=\alpha v_0$ with an arbitrary real scalar $\alpha \neq 0$,
\begin{equation}\label{eq:step1:proof}
    \|A^k x_0\| \geq c_0 |\alpha| \rho^k k^{d-1},
\end{equation}
for all sufficiently large $k\in\mathbb{N}$ and for some constant $c_0>0$, where $d$ is the size of the largest Jordan block associated with eigenvalues $\lambda$ satisfying $|\lambda| = \rho:=\rho(A)$ in the complex Jordan form of $A$. Since $A$ is real, it admits a \textit{real Jordan decomposition}. That is, there exists a real invertible matrix $P \in \mathbb{R}^{n\times n}$ and a block-diagonal matrix $J$ such that $A = P J P^{-1}$ with $J = \operatorname{diag}(J_1, \dotsc, J_L)$, and each $J_\ell$, $\ell=1,\ldots,L$, has the following properties:

I) If $\lambda_\ell \in \mathbb{R}$ is a real eigenvalue, then $J_\ell = \lambda_\ell I + N\in \mathbb{R}^{d_\ell\times d_\ell},
$ where \(N\) is a nilpotent matrix with ones on the superdiagonal and zeros elsewhere, $d_\ell$ is the size of the Jordan block associated with $\lambda_\ell$, and we have,
    \[
 J_\ell^k =
\begin{pmatrix}
\lambda_\ell^k & \binom{k}{1} \lambda_\ell^{k-1} &  \cdots & \binom{k}{d-1} \lambda_\ell^{k-d+1} \\
0 & \lambda_\ell^k  & \cdots & \binom{k}{d-2} \lambda_\ell^{k-d+2} \\
\vdots & \vdots  & \ddots & \vdots \\
0 & 0 & \cdots & \lambda_\ell^k
\end{pmatrix},\,\,\text{for all}\,k\in\mathbb{N}.
    \]

II) If $\lambda_\ell = \alpha_\ell + \mathrm{i}\beta_\ell$ is a complex eigenvalue with $\beta_\ell \neq 0$, then
    \[
    J_\ell = \begin{bmatrix}
    C_\ell & I_2 & & \\
    & C_\ell & \ddots & \\
    & & \ddots & I_2 \\
    & & & C_\ell
    \end{bmatrix}
     \in \mathbb{R}^{2d_\ell \times 2d_\ell},\quad C_\ell = \begin{bmatrix}
    \alpha_\ell & \beta_\ell \\
    -\beta_\ell & \alpha_\ell
    \end{bmatrix},
    \]
    where $I_2$ is the identity matrix of size~$2$. In this case, $d_\ell$ refers to the size of the original Jordan block over $\mathbb{C}$. We then have,
 \[
J_\ell^k =
\begin{pmatrix}
C_\ell^k & \binom{k}{1} C_\ell^{k-1} & \binom{k}{2} C_\ell^{k-2} & \cdots & \binom{k}{d-1} C_\ell^{k-d+1} \\
0 & C_\ell^k & \binom{k}{1} C_\ell^{k-1} & \cdots & \binom{k}{d-2} C_\ell^{k-d+2} \\
0 & 0 & C_\ell^k & \cdots & \binom{k}{d-3} C_\ell^{k-d+3} \\
\vdots & \vdots & \vdots & \ddots & \vdots \\
0 & 0 & 0 & \cdots & C_\ell^k
\end{pmatrix},
\]   
for all $k\in\mathbb{N}$. Let $\rho = |\lambda_{\ell_0}|$ correspond to some eigenvalue $\lambda_{\ell_0}$ and associated block $J_{\ell_0}$ of size $d=d_{\ell_0}$. We now construct the real initial vector $x_0$ depending on whether $\lambda_{\ell_0}$ is real or complex.

\smallskip
\noindent \emph{Case 1: $\lambda_{\ell_0} \in \mathbb{R}$ (real eigenvalue).} Let $e_d \in \mathbb{R}^{d}$ denote the standard basis vector $e_d = (0, 0, \dotsc, 0, 1)^\top,$ where the $1$ appears in the $d$-th coordinate. Define $\widetilde{e} \in \mathbb{R}^n$ to be the vector obtained by inserting $e_d$ into the coordinates corresponding to block $\ell_0$ and zeros elsewhere. Define $v_0 = P \widetilde{e}$. For any $\alpha \neq 0$, consider
\[
A^k  x_0 = \alpha A^k  v_0=\alpha P J^k P^{-1}v_0= \alpha P J^k\widetilde{e}.
\]
Since $J$ is block-diagonal and $\widetilde{e}$ is nonzero only in block $\ell_0$,
\[
J^k \widetilde{e} = \left( 0, \dotsc, 0, J_{\ell_0}^k e_d, 0, \dotsc, 0 \right)^\top.
\]
Applying $e_d$ to the powers of $J_{\ell_0}^k$, we observe that
\[
J_{\ell_0}^k e_d = \begin{pmatrix}
\binom{k}{d-1} \lambda_{\ell_0}^{k-d+1}  & \binom{k}{d-2} \lambda_{\ell_0}^{k-d+2} & \cdots & \lambda_{\ell_0}^k
\end{pmatrix}^\top.
\]
Thus,
\[
\|J^k \widetilde{e}\|=\|J_{\ell_0}^k e_d\| \geq  \bar c_0 |\lambda_{\ell_0}|^k k^{d-1}= \bar c_0 \rho^k k^{d-1},
\,
\,\text{as} \,
k \to \infty,
\]
for some constant $\bar c_0 > 0$. Since $P$ is invertible,
\[
\|A^k  x_0\| \!=\! |\alpha| \|P J^k \widetilde{e}\| \geq |\alpha| \frac{1}{\|P^{-1}\|} \|J_{\ell_0}^k e_d\|\geq  c_0|\alpha| \rho^k k^{d-1},
\]
for some $c_0>0$ depending on $P$ and $\bar c_0$.

\smallskip
\noindent \emph{Case 2: $\lambda_{\ell_0} \in \mathbb{C} \setminus \mathbb{R}$ (complex eigenvalue).} Suppose $\lambda_{\ell_0} = \alpha_{\ell_0} + \mathrm{i}\beta_{\ell_0}$, with $\rho = \sqrt{\alpha_{\ell_0}^2 + \beta_{\ell_0}^2}$. Let $e_{d-1,d} \in \mathbb{R}^{2d}$ be the vector
$
e_{d-1,d} = (0, \dotsc, 0, 1, 1)^\top,
$
where the last two entries are $1$. Define $\widetilde{e} \in \mathbb{R}^n$ by embedding $e_{d-1,d}$ into block $\ell_0$, zeros elsewhere similar to Case 1 and define $v_0 = P \widetilde{e}$. The block $C_{\ell_0}$ satisfies
\[
C_{\ell_0}^k = \rho^k \begin{bmatrix}
\cos(k\theta_{\ell_0}) & \sin(k\theta_{\ell_0}) \\
-\sin(k\theta_{\ell_0}) & \cos(k\theta_{\ell_0})
\end{bmatrix},
\]
where $\theta_{\ell_0} = \arg(\lambda_{\ell_0})$. Thus,
\begin{align*}
    J_{\ell_0}^k e_{d-1,d}= \Bigg(&
\binom{k}{d-1} \rho^{k-d+1} \begin{bmatrix}
\cos(k\theta_{\ell_0}) +\sin(k\theta_{\ell_0}) \\
\cos(k\theta_{\ell_0})-\sin(k\theta_{\ell_0})
\end{bmatrix}^\top, \\ & \cdots , \rho^{k} \begin{bmatrix}
\cos(k\theta_{\ell_0}) +\sin(k\theta_{\ell_0}) \\
\cos(k\theta_{\ell_0})-\sin(k\theta_{\ell_0})
\end{bmatrix}^\top
\Bigg)^\top,
\end{align*}
and we get, $\|J_{\ell_0}^k e_{d-1,d}\| \geq  \bar c_0 \rho^k k^{d-1},$
as $k \to \infty,$ for some $\bar c_0>0$ and thus,~\eqref{eq:step1:proof} is obtained.

\smallskip
\textbf{Step 2:} We then show that there exists a constant $c_1 > 0$ such that $\sum_{i=0}^{k-1} \|A^i\| \leq  c_1 \rho^k k^{d-1}$, as $k \to \infty,$
for all sufficiently large $k\in\mathbb{N}$. We know that,
$
A^i = P J^i P^{-1},
$
and taking norms, we have
\[
\|A^i\| \leq \|P\| \|J^i\| \|P^{-1}\| = \kappa(P) \|J^i\|,
\]
where $\kappa(P) = \|P\| \|P^{-1}\|$ is the condition number of $P$. The norm of $J_\ell^i$ satisfies $\|J_\ell^i\| \leq \bar c_0 \rho^i i^{d_\ell-1}$ for large $i$ and for some constant $\bar c_0 $. Since $J$ is block diagonal, we have
$
\|J^i\| = \max_\ell \|J_\ell^i\|.
$
For small $i$, the maximum may be achieved by a block with smaller eigenvalue but larger Jordan block size. 
However, for large $i$, the exponential growth of the blocks corresponding to eigenvalues with modulus $\rho$ dominates the polynomial growth. Thus, for large $i$,
$
\|J^i\| \leq \bar c_0 \rho(A)^i i^{d-1},
$
where $d$ is the size of the largest Jordan block associated with an eigenvalue of modulus $\rho$. Summing over $i$ from $0$ to $k-1$, we get
\[
\sum_{i=0}^{k-1} \|A^i\| \leq \kappa(P) \sum_{i=0}^{k-1} \|J^i\|\leq  c_1\rho(A)^k k^{d-1},
\]
for all large enough $k$. Note that, for small $i$, the terms $\|J^i\|$ are bounded and contribute a finite amount to the sum. Since $\rho> 1$,  the geometric growth eventually dominates, and the behavior of the sum is governed by the terms with large $i$.

\smallskip
\textbf{Step 3:} We then show that there exists an initial condition $x_0$ and an index $k_0$ such that for all $M>0$, $\|x_k\|\geq M$ holds for all $k\geq k_0$ with positive probability. By Chebyshev’s inequality in Euclidean form, for any $\beta>0$,
\begin{equation*}
   \mathbb{P}\Big(\|x_k-\mathbb{E}[x_k]\|<\beta \sqrt{\mathrm{tr}(\Sigma_k)} \Big)\geq 1-\frac{1}{\beta^2},
\end{equation*}
where $\Sigma_k = \operatorname{Cov}(x_k)$. Using the triangle inequality, $\|x_k-\mathbb{E}[x_k]\|\geq \|\mathbb{E}[x_k]\|-\|x_k\|$, we obtain the following inequality:
\begin{align*}
 \, &\mathbb{P}\Big (\|x_k\|>\|\mathbb{E}[x_k]\|-\beta \sqrt{\mathrm{tr}(\Sigma_k)} \Big )\geq 1-\frac{1}{\beta^2}.
\end{align*}
Covariance $\Sigma_k$ satisfies $ \Sigma_k = \sum_{i=0}^{k-1}  M_i$ with $M_i:=A^i B \Sigma_w B^\top (A^i)^\top$. Hence,
\begin{equation*}
    \operatorname{tr}(\Sigma_k) = \sum_{i=0}^{k-1} \operatorname{tr}(M_i)=\sum_{i=0}^{k-1} \|A^i B \Sigma_w^{1/2}\|_F^2\leq c^2\sum_{i=0}^{k-1} \|A^i\|_2^2, 
\end{equation*}
with $c=\|B \Sigma_w^{1/2}\|_F$. Moreover, taking expectation of the system dynamics yields $\mathbb{E}[x_{k+1}]=A\,\mathbb{E}[x_k]$, and it implies $\mathbb{E}[x_k]=A^kx_0$. Thus, using Steps $1$ and $2$,
\begin{align*}
    &\|\mathbb{E}[x_k]\|-\beta  \sqrt{\mathrm{tr}(\Sigma_k)} \geq \|A^kx_0\|-\beta c\sqrt{\sum_{i=0}^{k-1} \|A^i\|_2^2} \\ &\qquad\qquad\geq c_0 |\alpha| \rho^k k^{d-1}-\beta c c_1\rho^k k^{d-1} \geq c_2\rho^k k^{d-1},
\end{align*}
for some constants $c_1, c_0>0$, and $d\geq 1$ and all $\alpha\neq 0$.  Now choose $\beta>1$ and $\alpha\geq{(c_2+\beta cc_1)}/{c_0}$ ensuring the last inequality. Then the term $c_2\rho^k k^{d-1}$ grows unboundedly with $k$ and exceeds any given $M>0$ for $k$ large enough. Therefore, $\|x_k\|\geq M$ holds with the probability at least $1-\frac{1}{\beta^2}>0$.

\smallskip
\textbf{Step 4:} Finally, we invoke results from~\cite{meyn2012markov} to conclude that no drift function can satisfy condition~\textbf{V1}. Recall that a system is said to be \emph{non-evanescent} if, for all initial conditions $x_0$, we have $\mathbb{P}\left(\|x_k\|\to \infty \right)=0$. However, the result in Step~3 establishes that the system violates the non-evanescence condition. Therefore, by Theorem 9.4.1 in~\cite{meyn2012markov}, it follows that there exists no function $V$ and a compact set $C$ satisfying the drift condition~\eqref{eq:drift}.
\end{proof}

Theorem~\ref{theorem:linear} shows that when $\rho(A)<1$, quadratic drift and variant functions always exist and therefore the system almost surely reaches any target set containing the origin.  When $\rho(A) > 1$, divergence can be formally established, and no certificate satisfying the drift condition exists. These results offer  a complete characterization for systems with non-critical dynamics. In the following, we address the boundary case $\rho(A) = 1$, where the system exhibits more subtle behavior. In this setting, the spectral radius, Jordan structure, number of eigenvalues with a modulus of one, and noise directions, play a central role and require more delicate arguments.

We first show that 
the system diverges when there is a Jordan block of size $d \geq 2$  associated with an eigenvalue $\lambda$ such that $|\lambda| = 1$ even if $\rho(A)=1$.
The proof is similar to that of 
Theorem~\ref{theorem:linear} for $\rho(A)>1$.

\begin{proposition}\label{cor:rho_equals_1_jordan}
    Consider the system described by~\eqref{eq:linear_system} with $\rho(A)= 1$, and suppose the Jordan canonical form of $A$ contains a Jordan block of size greater than one associated with an eigenvalue $\lambda$ such that $|\lambda| = 1$. Then the system does not almost surely reach any bounded open target set.
\end{proposition}
\begin{proof} 
The proof follows the same reasoning of Theorem~\ref{theorem:linear} for $\rho(A)>1$; the only difference is that $\rho(A)= 1$, and we get the polynomial increment in the form of  $\|A^k x_0\| \geq c_0 |\alpha| k^{d-1}$. Thus, we can always find a real initial condition such that the system diverges with positive probability. Similar to Step~4 in the proof of Theorem~\ref{theorem:linear}, we can show that the system does not admit a valid drift function and no bounded open target set is almost surely reachable by the system.
    \end{proof}

In the following, we provide two fundamental theorems and one lemma that will be used in the proof of the main theorem for the almost sure reachability of linear systems with $\rho(A)=1$.
 \begin{theorem}[Lindeberg-Feller central limit theorem]\label{thm:LFCLT} Suppose for each $k\in\mathbb{N}_{>0}$, the random variables $Y_{k,j}\in\mathbb{R}$ for $j=0,1,\ldots, k-1$, are independent, zero mean $\mathbb{E}[Y_{k,j}]=0$, and variances $\operatorname{Var}(Y_{k,j})=s^2_{k,j}$  (forming a triangular array). Define the total variance $s^2_k:=\sum_{j=0}^{k-1}s^2_{k,j}$. Assume $0<s^2_k<\infty$ for all large $k$. Then if the following \textit{Lindeberg condition} holds:
\begin{equation*}
\lim_{k \to \infty} \frac{1}{s_k^2} \sum_{j=0}^{k-1} \mathbb{E}\left[Y_{k,j}^2  \mathbf{1}_{\{|Y_{k,j}| > \varepsilon s_k\}}\right] = 0,
\end{equation*}
for any $\varepsilon>0$, the sequence of sums ${Y_k}:=\sum_{j=0}^{k-1}Y_{k,j}$ converges in distribution to a normal distribution with mean zero and variance $s^2_k$. 
\end{theorem}
\begin{proof}
    See e.g., Theorem 5.33 in~\cite{boos2013essential}.
\end{proof}
\begin{theorem}[Cramér--Wold]\label{thm:CW}
Let \( \{X_k\}_{k=1}^\infty \) be a sequence of $\mathbb{R}^n$-valued random vectors, and let \( X \in \mathbb{R}^n \) be a random vector. Then \( X_k \xrightarrow{d} X \) (i.e., \( X_k \) converges in distribution to \( X \)) if and only if
$
u^\top X_k \xrightarrow{d} u^\top X, 
$ for all $u \in \mathbb{R}^n$.
\end{theorem}
\begin{proof}
    See e.g., Theorem 5.31 in~\cite{boos2013essential}.
\end{proof}

\begin{lemma}[Borel--Cantelli]
Let \( \{x_k\}_{k=1}^\infty \) be a sequence of random vectors in \( \mathbb{R}^n \). If
$
\sum_{k=1}^\infty \mathbb{P}(x_k \in \mathcal{B}) < \infty,
$
for a bounded set \( \mathcal{B} \subset \mathbb{R}^n \), then
$
\mathbb{P}(x_k \in \mathcal{B}, \text{ infinitely often}) = 0
$.
\end{lemma}
\begin{proof}
    See e.g.,~\cite{chandra2012borel}.
\end{proof}

To proceed with our analysis for linear systems with $\rho(A)=1$, we make the following assumption on system~\eqref{eq:linear_system}. A more detailed explanation will follow later in the section.

 \begin{assumption}\label{Assum:rho1}
 The matrix $A$ has spectral radius $\rho(A) =1$, and all Jordan blocks of $A$ associated with eigenvalues $\lambda$ satisfying $|\lambda| = 1$ are of size one. Moreover, the noise process satisfies $\mathbb{E}[\|w_k\|^3] < \infty$ and the matrix $B$ is full rank with $n=m$. 
 \end{assumption}
 
\begin{theorem}\label{theorem:rho1:1}
Under Assumption~\ref{Assum:rho1}, almost sure reachability for the linear  system~\eqref{eq:linear_system} holds for any bounded open target set containing the origin if and only if the number of eigenvalues with modulus one is at most two.
\end{theorem}
\begin{proof}
    \textbf{Part A:} First, we assume that all eigenvalues of $A$ lie on the unit circle. In this case, the system dimension $n$ equals the number of eigenvalues with modulus one. We now prove the result for the case $n\leq 2$.

\textbf{Part A-I: $n \leq 2$ — Almost Sure Reachability.} We prove the result by verifying both \textbf{V1} and \textbf{V2}. For \textbf{V1}, we consider two cases based on the system size.

\textbf{Case $n = 1$.} 
In this case, $A \in \{-1, 1\}$, $B\neq 0$ and $|Ax| = |x|$ holds. Define $V(x) := \sqrt{\ln |x|}$ for $|x|\geq 1$. Since~\eqref{eq:drift} is required to hold only outside a compact set $C$, restricting the definition of $V(x)$ to the domain $|x|\geq 1$ does not affect the validity of the results. A Taylor expansion of $V(Ax + Bw)$ around $w = 0$ gives:
\begin{equation*}
\mathbb{E}[V(Ax + Bw)\mid x] \!=\! V(Ax) \!+\! \frac{1}{2} \mathbb{E}[w^2] \frac{\mathrm{d}^2 V}{\mathrm{d}x^2} + \mathbb{E}[R(x,w)\mid x],
\end{equation*}
where $R(x,w)$ is the third-order remainder term and 
\begin{equation*}
\frac{\mathrm{d}^2 V}{\mathrm{d}x^2} = \frac{-B^2}{2x^2(\ln |x|)^{1/2}} \left( \frac{1}{2\ln |x|} + 1 \right),
\quad x \ne 0.
\end{equation*}
For all $|x|>1$, this expression is negative. Since the third-order remainder satisfies $|\mathbb{E}[R(x, w)\mid x]| = O(|x|^{-3})$ under Assumption~\ref{Assum:rho1}, and noting that $\Sigma_w=\mathbb{E}[w^2]>0$ is finite, we have $\mathbb{E}[V(A x+Bw) \mid x] \leq  V(x)$ for sufficiently large $|x|$, and the drift condition \textbf{V1} holds.

\smallskip
\textbf{Case $n = 2$.} Since $A$ is diagonalizable with eigenvalues on the unit circle, there exists an invertible $P \in \mathbb{R}^{2 \times 2}$ such that $A = P R_\theta P^{-1},$ where $R_\theta$ is the rotation matrix. Define the matrix $Q := (P P^\top)^{-1}$. Then $Q \succ 0$ and $A^\top Q A = Q$. Define the norm $\|x\|_\star := \sqrt{x^\top Q x}$, so $\|Ax\|_\star = \|x\|_\star$ for all $x \in \mathbb{R}^2$. Fix any $x$, let $z := Ax$, so $\|z\|_\star = \|x\|_\star =: r > 0$. Define $V(x) := \sqrt{\ln \|x\|_\star}$ for all $x$ satisfying $\|x\|_\star\geq 1$, and then define:
\begin{equation*}
g(w) := V(z + Bw) = \sqrt{ \ln \|z + Bw\|_\star },
\end{equation*}
for any fixed $z$. Using Taylor expansion of $g(w)$ around $Bw = 0$, taking expectations and noting $\mathbb{E}[w] = 0$, we get:
\begin{equation*}
\mathbb{E}[g(w)] = g(0) + \frac{1}{2} \mathbb{E}[w^\top \nabla^2 g(0) w] + \mathbb{E}[R(w)],
\end{equation*}
with
\begin{align*}
\nabla g(0) = \frac{1}{2 \sqrt{\ln r}} \frac{B^\top Q z}{r^2},\,\,\, &\nabla^2 g(0) = \frac{1}{2r^2 \sqrt{\ln r}} \Bigg( {B^\top QB}\\
 &- \frac{2 B^\top Q z z^\top Q B}{r^2} \Bigg)- \frac{B^\top Q z z^\top Q B}{4r^4 (\ln r)^{3/2}}.
\end{align*}
Let $T_1 := \operatorname{tr}(Q B \Sigma_w B^\top)>0$ and $T_2 := z^\top Q B \Sigma_w B^\top Q z = x^\top A^\top Q B \Sigma_w B^\top Q A x$. Then,
\begin{equation}\label{eq:exp:second}
\mathbb{E}[w^\top \nabla^2 f(0) w]
= \frac{1}{ r^2\sqrt{\ln r}} \left( \frac{T_1}{2} - \frac{ T_2}{r^2} \right)
- \frac{T_2}{4 r^4(\ln r)^{3/2}}.    
\end{equation}
To analyze the sign, define $y := Q^{1/2} A x$, and $\tilde{\Sigma}_w := Q^{1/2} B \Sigma_w B^\top Q^{1/2} \succ 0$, then,
\begin{equation*}
T_2 = y^\top \tilde{\Sigma}_w y, \quad
r^2 = \|x\|_\star^2 = \|Q^{1/2} x\|^2.
\end{equation*}
Note that $\|y\|^2 = \|Q^{1/2} A x\|^2 = \|x\|_\star^2 = r^2$. Let $\lambda_1, \lambda_2 > 0$ be the eigenvalues of $\tilde{\Sigma}_w$, with orthonormal eigenvectors $v_1, v_2$. Writing $y = a v_1 + b v_2$, we have:
\begin{align}\label{eq:ineq:proof:d2}
\frac{T_2}{r^2} = & \frac{\lambda_1 a^2 + \lambda_2 b^2}{a^2 + b^2} \geq \frac{\lambda_1 + \lambda_2}{2} = \frac{1}{2} \operatorname{tr}(\tilde{\Sigma}_w)\nonumber\\ = & \frac{1}{2} \operatorname{tr}(Q B \Sigma_w B^\top) = \frac{T_1}{2}.
\end{align}
Using this, we also bound the last term in~\eqref{eq:exp:second}:
\begin{equation*}
- \frac{T_2}{4 r^4(\ln r)^{3/2}} \leq - \frac{T_1}{8r^2(\ln r)^{3/2}}.
\end{equation*}
Hence:
\begin{equation*}
\mathbb{E}[w^\top \nabla^2 f(0) w] \leq - \frac{c_1}{r^2 \sqrt{\ln r}} - \frac{c_2}{r^2 (\ln r)^{3/2}},
\end{equation*}
for some constants $c_1 \geq 0$ and $c_2 > 0$. Therefore, even when the first term vanishes, the second is strictly negative. Additionally, the third-order remainder satisfies $|\mathbb{E}[R(w)]| \leq {c_3}/{r^3}$ for some $c_3 > 0$ because $\mathbb{E}[\|w\|^3]<\infty$. Combining all terms:
\begin{equation*}
\mathbb{E}[V(Ax + Bw) \mid x] \leq V(x) - \frac{c_1}{r^2 \sqrt{\ln r}} - \frac{c_2}{r^2 (\ln r)^{3/2}} + \frac{c_3}{r^3}.
\end{equation*}
Since $r \gg (\ln r)^{3/2}$ for large $r$, the right-hand side is strictly less than $V(x)$ for all sufficiently large $r$. Thus, the drift condition holds. Note that for $n > 2$, inequality~\eqref{eq:ineq:proof:d2} no longer holds in general.

\medskip
\noindent
For \textbf{V2}, we again use the defined weighted norm and propose the following quadratic variant function
\[
U(x) := \|x\|_\star^2 - b = x^\top Q x - b,
\]
for some constant $b > 0$ depending on $G$. For $n=1$, we use the standard Euclidean norm $|x|=\|x\|_\star$ and still $\|Ax\|_\star=\|x\|_{\star}$ holds with $Q=1$. First, for the defined $V(x)=\sqrt{\ln \|x\|_\star}$, we observe that the implication 
$
    V(x) \leq r \implies U(x) \leq H(r)
$
holds with $H(r)=\exp{(2r^2)}-b$. Since $B \in \mathbb{R}^{n \times n}$ is full rank and the noise $w$ with covariance $\Sigma_w\succ 0$ has a distribution with full support in an open neighborhood of the origin. Then $B w$ also has full support in a neighborhood of the origin in $\mathbb{R}^n$. As a result, for any $x$ with $\|x\|_\star^2>b$, there exists a constant $\delta > 0$ such that $\|Ax+Bw\|_\star^2-\|x\|_\star^2\leq - \delta$ holds with positive probability for some fixed $\delta > 0$. Therefore, $U$ satisfies the variant condition \textbf{V2}, completing the proof of almost sure reachability when $n \leq 2$.

\medskip
\noindent
\textbf{Part A-II: $n > 2$ — Non-reachability.} The proof unfolds in the following key steps:
\begin{enumerate}[label=\alph*)]
    \item Establish independence of a scalar projection by exploiting independence of the noise sequence.
    \item Apply the Lindeberg--Feller central limit theorem to show that the projection converges to a Gaussian.
    \item Use the Cramér--Wold theorem to lift this to a multivariate central limit result.
    \item Bound probabilities of hitting any compact ball and apply the Borel--Cantelli lemma to show that the state diverges almost surely.
\end{enumerate}

\paragraph{Independence}
To establish non-reachability, it suffices to show that the system fails to meet the conditions for a single initial condition. In particular, we consider $x_0=0$ without loss of generality. The solution for~\eqref{eq:linear_system} becomes $x_k = \sum_{j=0}^{k-1} A^j B w_{k-1-j}$ and we define $ X_k := \frac{1}{\sqrt{k}} x_k$. Fix any $u \in \mathbb{R}^n$ and define
\begin{equation*}
Y_{k,j} := \frac{1}{\sqrt{k}} u^\top A^j B w_{k-1-j}, \quad 0 \leq j\leq k-1.
\end{equation*}
Then $u^\top X_k = \sum_{j=0}^{k-1} Y_{k,j}=:Y_{k}$. For each $k$, the entries $\{Y_{k,j}\}_{j=0}^{k-1}$ are independent as they are obtained by measurable transformations of independent random variables $w_{k-j-1}$. Moreover, they are centered at zero, i.e., $\mathbb{E}[Y_{k,j}]=0$ for all $j$. 

\paragraph{Applying Lindeberg–Feller Theorem}
We now compute the variance of $Y_{k,j}$, which is given by
\begin{equation*}
\operatorname{Var}(Y_{k,j}) = \frac{1}{k} u^\top A^j B \Sigma_w B^\top (A^j)^\top u.
\end{equation*}
Since $B$ is full rank, $B \Sigma_w B^\top \succ 0$ and $\operatorname{Var}(Y_{k,j})>0$ for all $u\neq 0$ and all $j$ and $k$. We then define the total variance of the $k$-th row as
\begin{align*}
s_k^2 :=& \sum_{j=0}^{k-1} \operatorname{Var}(Y_{k,j}) = \frac{1}{k} \sum_{j=0}^{k-1} u^\top A^j  B \Sigma_w B^\top (A^j)^\top u
 \\ =&\frac{1}{k} \sum_{j=0}^{k-1} \| B^\top (A^j)^\top u \|_{\Sigma_w}^2
\!\leq\! \frac{1}{k} \sum_{j=0}^{k-1} \|B\|_{\Sigma_w}^2 \|A^j\|_{\Sigma_w}^2 \|u\|_{\Sigma_w}^2.
\end{align*}
Since $\rho(A)=1$, all Jordan blocks of $A$ are size one, and $\Sigma_w\succ 0$ is finite, it follows that $0<c\leq s_k^2 \leq C<\infty$ (see Step 2 in the proof of Theorem~\ref{theorem:linear} for $\rho(A)>1$). We aim to apply the Lindeberg–Feller central limit theorem as stated in Theorem~\ref{thm:LFCLT}. To this end, we first check the Lindeberg condition. Observe that
\begin{equation*}
|Y_{k,j}| \leq \frac{1}{\sqrt{k}} \|u^\top A^j B\|  \|w_{k-j-1}\| \leq \frac{c_0}{\sqrt{k}} \|w_{k-j-1}\|,
\end{equation*}
for some $c_0>0$. Thus,
\begin{equation*}
|Y_{k,j}| > \varepsilon s_k \Rightarrow \|w_{k-j-1}\| > R_k := \frac{\varepsilon s_k \sqrt{k}}{c_0}\geq \frac{\varepsilon \sqrt{ck}}{c_0},
\end{equation*}
with $\lim_{k\to\infty} R_k=\infty$. Then,
\begin{equation*}
Y_{k,j}^2  \mathbf{1}_{\{|Y_{k,j}| > \varepsilon s_k\}} \leq \frac{c_0^2}{k}  \|w_{k-j-1}\|^2  \mathbf{1}_{\{ \|w_{k-j-1}\| > R_k \}}.
\end{equation*}
Taking expectation and summing,
\begin{equation*}
\sum_{j=0}^{k-1} \mathbb{E}[Y_{k,j}^2  \mathbf{1}_{\{|Y_{k,j}| > \varepsilon s_k\}}]
\!\leq \!{c_0^2}  \mathbb{E}[\|w_0\|^2  \mathbf{1}_{\{\|w_0\| > R_k\}}]\!=:L_k,
\end{equation*}
since $\mathbb{E}[\|w_i\|^2  \mathbf{1}_{\{\|w_i\| > R_k\}}]=\mathbb{E}[\|w_j\|^2  \mathbf{1}_{\{\|w_j\| > R_k\}}]$ for any $i$ and $j$, because $w_i$s are drawn from an identical distribution. One can then observe that $\lim_{k\to\infty}L_k=0$, since $\Sigma_w=\mathbb{E}[\|w_0\|^2] < \infty$ and $\lim_{k\to\infty} R_k=\infty$. Therefore, the Lindeberg condition holds. By the Lindeberg--Feller central limit theorem $Y_k\in\mathbb{R}$ converges to a normal distribution when $k\to\infty$. 

\paragraph{Applying Cramér–Wold Theorem}
In order to establish convergence result for vector $X_k\in \mathbb{R}^n$, we can apply the Cramér--Wold theorem~\ref{thm:CW}, and it reads,
$
X_k \xrightarrow{d} \mathcal{N}(0, \Gamma),
$
for some positive definite $\Gamma$, because $u^\top X_k =Y_{k}$ for any $u$. 

\paragraph{Applying Borel–Cantelli Lemma}
For any fixed ball $\mathcal{B} \subset \mathbb{R}^n$,
\begin{align*}
\mathbb{P}(x_k \in \mathcal B) &= \mathbb{P}\left(X_k \in \frac{1}{\sqrt{k}} \mathcal B\right) =\int_{\frac{1}{\sqrt{k}} \mathcal{B}} f_k(z) \, \mathrm{d}z\\ &\leq  \sup_{z \in \frac{1}{\sqrt{k}} \mathcal{B}} f_k(z) \cdot \operatorname{Vol}\left(\tfrac{1}{\sqrt{k}} \mathcal{B}\right)\leq  \frac{\bar c}{k^{n/2}}\operatorname{Vol}\left(\mathcal{B}\right),
\end{align*}
for some constant $\bar c>0$, since $f_k(z)$ converges to $\mathcal{N}(0, \Gamma)$ pointwise, where $f_k(z)$ is the density of $X_k$. Since $n > 2$, it follows that $\sum_k \mathbb{P}(x_k \in \mathcal B) < \infty$. 
Hence, by the Borel--Cantelli Lemma,
\begin{equation*}
\mathbb{P}(x_k \in \mathcal B\,, \text{infinitely often}) \!=\! 0,
\end{equation*}
for any fixed ball $\mathcal{B} \subset \mathbb{R}^n$ and then the system fails to almost surely reach any bounded open target set.

\textbf{Part B:} Then, we generalize the results to the case where $\rho(A)=1$, allowing for eigenvalues with modulus less than one. We define the unit-circle invariant subspace of $A$ as
\begin{equation*}
E_A := \operatorname{span} \left\{ \operatorname{Re}(v),\ \operatorname{Im}(v) \mid v\in \mathbb{C}^n, A v = \lambda v,\ |\lambda| = 1 \right\},
\end{equation*}
that is, $E_A\subseteq R^n$ is the real subspace spanned by the real and imaginary parts of all (possibly complex) eigenvectors associated with eigenvalues of modulus one. Note that $\dim(E_A)$ is the number of eigenvalues of $A$ such eigenvalues. Let us decompose the state space into the invariant subspace $E_A$, and its complement, which is associated with eigenvalues strictly inside the unit circle. The dynamics restricted to $E_A$ behave like a system with unit-modulus eigenvalues and Jordan blocks of size one, while the complementary dynamics decay over time due to strict stability. From part A of this proof, we know that when the dynamics are entirely supported on $E_A$ and $\dim(E_A) \leq 2$, any bounded open target set is almost surely reachable. Meanwhile, the decaying part of the dynamics (associated with $|\lambda| < 1$) always leads to reachability, as shown in Theorem~\ref{theorem:linear} for $\rho(A)<1$. Hence, when $\dim(E_A) \leq 2$, the full system exhibits almost sure reachability.

On the other hand, if $\dim(E_A) > 2$, then by Theorem~\ref{theorem:rho1:1}, the restriction of the system to $E_A$ lacks almost sure reachability, regardless of the eigenvalues with modulus less than one. Consequently, the full system does not reach any bounded open target set with probability one.
\end{proof}
Theorem~\ref{theorem:rho1:1} provides a principled generalization of classical random walk behavior to stochastic linear systems governed by~\eqref{eq:linear_system} and it is aligned with classical recurrence results for random walks when $A = I$ (e.g., the Chung–Fuchs Theorem~\cite{chung1951distribution}). In addition, the framework provides explicit certificates as developed in~\cite{majumdar2024necessary}: a logarithmic form for the drift function and a quadratic form for the variant function, which together verify almost sure reachability of linear systems with $\rho(A)=1$, under the assumptions that $A$ has at most two eigenvalues with modulus one, all corresponding Jordan blocks are of size one, and $B$ is full rank. 

While standard random walks are sums of i.i.d. variables amenable to the classical central limit theorem, the presence of the matrix $A$ introduces temporal dependence. To address this, we invoke the Lindeberg–Feller central limit theorem for triangular arrays and apply the Cramér–Wold theorem to lift convergence to the multivariate setting. This establishes asymptotic Gaussian behavior despite the system’s lack of independence or stationarity.

In Theorem~\ref{theorem:rho1:1}, we assume $B \in \mathbb{R}^{n \times n}$ is full rank, as stated in Assumption~\ref{Assum:rho1}, ensuring that noise affects all directions in the state space. This non-degeneracy is essential for both the central limit analysis and reachability guarantees.

\begin{Example}
Consider~\eqref{eq:linear_system} with $A = I_2$ and $B = [1,1]^\top$. The noise acts only along $[1,1]^\top$, a one-dimensional subspace. Let $y = P x$ with $P \in \mathbb{R}^{2 \times 2}$ be defined by $P_{1,1} = 1$, $P_{1,2} = 0$, $P_{2,1} = -1$, and $P_{2,2} = 1$, so that in the new coordinates the system becomes $y_{k+1} = y_k + [1,0]^\top w_k$. The second coordinate evolves deterministically. In this transformed system, any initial condition of the form $y_0 = [0\,,\alpha]^\top$ remains of the form $y_k=[\beta_k\,, \alpha]^\top$ for all $k$. By choosing $\alpha$ large enough so that $\alpha>\sup_{x\in G}\|x\|$ for any bounded open target set $G$ containing the origin, it can be seen that $y_0\notin G$ because $\|y_0\|=\alpha>\sup_{x\in G}\|x\|$. One then ensures that $y_k\notin G$ for all $k$, because $\|y_k\|\geq \|y_0\|$ for all $k$. Therefore, reachability fails in this degenerate case, highlighting the necessity of the full-rank assumption on $B$.
\end{Example}

The divergence result $\|x_k\| \to \infty$ also relies critically on $B$ being full rank. Without this, some directions may remain unaffected, and the state norm may remain bounded.

\begin{Example}
Consider~\eqref{eq:linear_system} with $A = I_3$ and $B = \operatorname{diag}(1, 1, 0)$. The first two coordinates behave as a two-dimensional random walk, while the third remains constant. Since the norm of a two-dimensional random walk does not diverge almost surely, the overall state norm $\|x_k\|$ also remains bounded.
\end{Example}

The reachability framework based on Theorem~\ref{thm:nece_suff} identifies two distinct failure modes from Theorems~\ref{theorem:linear} and~\ref{theorem:rho1:1}. The first is structural: if $\rho(A)>1$ or $\rho(A)=1$ with $A$ containing Jordan blocks of size greater than one or more than two unit-modulus eigenvalues, divergence occurs with positive probability and no drift function satisfying \textbf{V1} can exist. The second is due to insufficient excitation: when $B$ is not full rank, the noise fails to affect all directions, preventing consistent decrease of the variant function and violating \textbf{V2}. Thus, reachability may fail either due to divergence or noise degeneracy.

\section{Conclusion}
\label{sec:conclusion}
This work advances the theoretical understanding of almost sure reachability in discrete-time stochastic systems. Building on the recently established drift and variant certificates, which provide a necessary and sufficient condition, we investigated when such certificates can actually be constructed within commonly used function classes. Our results show that restricting to fixed templates, such as polynomial or quadratic functions, may lead to a loss of completeness: we gave explicit examples of polynomial systems that do not admit polynomial certificates, and linear systems that fail to admit quadratic ones, despite satisfying almost sure reachability. We then focused on linear systems with additive stochastic disturbances and developed a complete structural characterization for almost sure reachability based on system properties. We showed that reachability can be determined by spectral properties of the system matrix, such as the spectral radius, Jordan block structure, and system dimension. In particular, we identified the precise conditions under which quadratic certificates are sufficient and constructed explicit non-quadratic certificates using logarithmic functions when needed. These results generalize classical random walk theory to more complex stochastic dynamical systems.

Beyond highlighting the limitations of polynomial-template-based approaches, this work lays the foundation for future research in several directions. One avenue is to develop computational methods for constructing certificates with formal guarantees, particularly in the polynomial setting. Another is to extend the framework toward control synthesis, where the goal is to design inputs that ensure almost sure reachability to target sets. Finally, a broader and more challenging direction is to seek structural characterizations for almost sure reachability in nonlinear systems, going beyond the linear regime explored in this paper.

 \section*{Acknowledgment}
 The authors would like to acknowledge S.V. Ramesh\,\orcidlink{0009-0006-5187-5415} at MPI-SWS for the discussions in the early stages of developing the results.

\bibliographystyle{IEEEtran}
\bibliography{ASReachLinear}

\begin{thebibliography}{10}
\providecommand{\url}[1]{#1}
\csname url@samestyle\endcsname
\providecommand{\newblock}{\relax}
\providecommand{\bibinfo}[2]{#2}
\providecommand{\BIBentrySTDinterwordspacing}{\spaceskip=0pt\relax}
\providecommand{\BIBentryALTinterwordstretchfactor}{4}
\providecommand{\BIBentryALTinterwordspacing}{\spaceskip=\fontdimen2\font plus
\BIBentryALTinterwordstretchfactor\fontdimen3\font minus \fontdimen4\font\relax}
\providecommand{\BIBforeignlanguage}[2]{{%
\expandafter\ifx\csname l@#1\endcsname\relax
\typeout{** WARNING: IEEEtran.bst: No hyphenation pattern has been}%
\typeout{** loaded for the language `#1'. Using the pattern for}%
\typeout{** the default language instead.}%
\else
\language=\csname l@#1\endcsname
\fi
#2}}
\providecommand{\BIBdecl}{\relax}
\BIBdecl

\bibitem{lygeros2004reachability}
J.~Lygeros, ``On reachability and minimum cost optimal control,'' \emph{Automatica}, vol.~40, no.~6, pp. 917--927, 2004.

\bibitem{bujorianu2012stochastic}
L.~M. Bujorianu, \emph{Stochastic reachability analysis of hybrid systems}.\hskip 1em plus 0.5em minus 0.4em\relax Springer Science \& Business Media, 2012.

\bibitem{baier2008principles}
C.~Baier and J.-P. Katoen, \emph{Principles of model checking}.\hskip 1em plus 0.5em minus 0.4em\relax MIT press, 2008.

\bibitem{blanchini2008set}
F.~Blanchini and S.~Miani, \emph{Set-theoretic methods in control}.\hskip 1em plus 0.5em minus 0.4em\relax Springer, 2008, vol.~78.

\bibitem{ames2019control}
A.~D. Ames, S.~Coogan, M.~Egerstedt, G.~Notomista, K.~Sreenath, and P.~Tabuada, ``Control barrier functions: Theory and applications,'' in \emph{2019 18th European control conference (ECC)}.\hskip 1em plus 0.5em minus 0.4em\relax Ieee, 2019, pp. 3420--3431.

\bibitem{prajna2007framework}
S.~Prajna, A.~Jadbabaie, and G.~J. Pappas, ``A framework for worst-case and stochastic safety verification using barrier certificates,'' \emph{IEEE Transactions on Automatic Control}, vol.~52, no.~8, pp. 1415--1428, 2007.

\bibitem{lavaei2022automated}
A.~Lavaei, S.~Soudjani, A.~Abate, and M.~Zamani, ``Automated verification and synthesis of stochastic hybrid systems: A survey,'' \emph{Automatica}, vol. 146, p. 110617, 2022.

\bibitem{kordabad2024control}
A.~B. Kordabad, M.~Charitidou, D.~V. Dimarogonas, and S.~Soudjani, ``Control barrier functions for stochastic systems under signal temporal logic tasks,'' in \emph{2024 European Control Conference (ECC)}.\hskip 1em plus 0.5em minus 0.4em\relax IEEE, 2024, pp. 3213--3219.

\bibitem{santoyo2021barrier}
C.~Santoyo, M.~Dutreix, and S.~Coogan, ``A barrier function approach to finite-time stochastic system verification and control,'' \emph{Automatica}, vol. 125, p. 109439, 2021.

\bibitem{chatterjee2016optimal}
K.~Chatterjee, M.~Chmelik, R.~Gupta, and A.~Kanodia, ``Optimal cost almost-sure reachability in {POMDP}s,'' \emph{Artificial Intelligence}, vol. 234, pp. 26--48, 2016.

\bibitem{jafarpour2025probabilistic}
S.~Jafarpour, Z.~Liu, and Y.~Chen, ``Probabilistic reachability analysis of stochastic control systems,'' \emph{IEEE Transactions on Automatic Control}, 2025.

\bibitem{mciver2017new}
A.~McIver, C.~Morgan, B.~L. Kaminski, and J.-P. Katoen, ``A new proof rule for almost-sure termination,'' \emph{Proceedings of the ACM on Programming Languages}, vol.~2, no. POPL, pp. 1--28, 2017.

\bibitem{abate2021learning}
A.~Abate, M.~Giacobbe, and D.~Roy, ``Learning probabilistic termination proofs,'' in \emph{Computer Aided Verification: 33rd International Conference, CAV 2021, Virtual Event, July 20--23, 2021, Proceedings, Part II 33}.\hskip 1em plus 0.5em minus 0.4em\relax Springer, 2021, pp. 3--26.

\bibitem{ferrer2015probabilistic}
L.~M. Ferrer~Fioriti and H.~Hermanns, ``Probabilistic termination: Soundness, completeness, and compositionality,'' in \emph{Proceedings of the 42nd Annual ACM SIGPLAN-SIGACT Symposium on Principles of Programming Languages}, 2015, pp. 489--501.

\bibitem{majumdar2025sound}
R.~Majumdar and V.~Sathiyanarayana, ``Sound and complete proof rules for probabilistic termination,'' \emph{Proceedings of the ACM on Programming Languages}, vol.~9, no. POPL, pp. 1871--1902, 2025.

\bibitem{majumdar2024necessary}
R.~Majumdar, V.~Sathiyanarayana, and S.~Soudjani, ``Necessary and sufficient certificates for almost sure reachability,'' \emph{IEEE Control Systems Letters}, 2024.

\bibitem{edition1999linear}
C.-T. Chen, \emph{Linear System Theory and Design}.\hskip 1em plus 0.5em minus 0.4em\relax Oxford University Press, 1999.

\bibitem{ahmadi2011globally}
A.~A. Ahmadi, M.~Krstic, and P.~A. Parrilo, ``A globally asymptotically stable polynomial vector field with no polynomial {L}yapunov function,'' in \emph{2011 50th IEEE Conference on Decision and Control and European Control Conference}.\hskip 1em plus 0.5em minus 0.4em\relax IEEE, 2011, pp. 7579--7580.

\bibitem{chung1951distribution}
K.~L. Chung and W.~H.~J. Fuchs, ``On the distribution of values of sums of random variables,'' \emph{Selected Works Of Kai Lai Chung}, pp. 157--168, 1951.

\bibitem{meyn2012markov}
S.~P. Meyn and R.~L. Tweedie, \emph{Markov chains and stochastic stability}.\hskip 1em plus 0.5em minus 0.4em\relax Springer Science \& Business Media, 2012.

\bibitem{boos2013essential}
D.~D. Boos and L.~A. Stefanski, \emph{Essential statistical inference: theory and methods}.\hskip 1em plus 0.5em minus 0.4em\relax Springer Science \& Business Media, 2013, vol. 120.

\bibitem{chandra2012borel}
T.~K. Chandra, \emph{The Borel-Cantelli Lemma}.\hskip 1em plus 0.5em minus 0.4em\relax Springer Science \& Business Media, 2012.

\end{thebibliography}
\end{document}